\documentclass[11pt, oneside]{article}  
\usepackage{amsthm}
\usepackage{graphicx}
\usepackage{amssymb}
\usepackage{amsmath}
\usepackage{algorithm2e}
\usepackage{psfrag}
\usepackage{pbox}
\usepackage{authblk}
\usepackage{multirow}
\usepackage{subfigure}
\newtheorem{theorem}{Theorem}
\newtheorem{lemma}{Lemma}

\graphicspath{{./fig/}}

\begin{document}

\title{Network Optimization on Partitioned Pairs of Points}
\date{}

\author[1]{Esther M. Arkin}
\author[2]{Aritra Banik}
\author[2]{Paz Carmi}
\author[3]{Gui Citovsky}
\author[1]{Su Jia}
\author[2]{Matthew J. Katz}
\author[1]{Tyler Mayer}
\author[1]{Joseph S. B. Mitchell}

\affil[1]{Dept. of Applied Mathematics and Statistics, Stony Brook University, Stony Brook USA.\\
\texttt{\{esther.arkin, su.jia, tyler.mayer, joseph.mitchell\}@stonybrook.edu}}
\affil[2]{Dept. of Computer Science, Ben-Gurion University, Beersheba Israel.\\
\texttt{\{aritrabanik,carmip\}@gmail.com, matya@cs.bgu.ac.il}}
\affil[3]{Google, Manhattan USA.\\
\texttt{gcitovsky@gmail.com}}

\maketitle

\begin{abstract}
Given $n$ pairs of points, $\mathcal{S} = \{\{p_1, q_1\}, \{p_2, q_2\}, \dots, \{p_n, q_n\}\}$, in some metric space, we study the problem of two-coloring the points within each pair, red and blue, to optimize the cost of a pair of node-disjoint networks, one over the red points and one over the blue points.  In this paper we consider our network structures to be spanning trees, traveling salesman tours or matchings. We consider several different weight functions computed over the network structures induced, as well as several different objective functions.  We show that some of these problems are NP-hard, and  provide constant factor approximation algorithms in all cases. 

\end{abstract}

\section{Introduction}

We study a class of network optimization problems on {\em pairs} of sites in a metric space. Our goal is to determine how to split each pair, into a ``red'' site and a ``blue'' site, in order to optimize {\em both} a network on the red sites and a network on the blue sites.
In more detail, given $n$ pairs of points, $\mathcal{S} = \{\{p_1, q_1\}, \{p_2, q_2\}, \dots, \{p_n, q_n\}\}$, in the Euclidean plane or in a general metric space, we define a {\em feasible coloring} of the points in $S = \bigcup_{i=1}^{n}\{p_i, q_i\}$ to be a coloring, $S = R \cup B$, such that $p_i \in R$ if and only if $q_i \in B$. Among all feasible colorings of $\mathcal{S}$, we seek one which optimizes the cost function over a {\em pair} of network structures, spanning trees, traveling salesman tours (TSP tours) or matchings,  one on the red set and one on the blue set. Let $f(X)$ be a certain structure computed on point set $X$ and let $\lambda(X)$ be the longest edge of a bottleneck  structure, $f(X)$, computed on point set $X$. For each of the aforementioned structures we consider the objective of (over all feasible colorings $S = R \cup B$) minimizing  $|f(R)| + |f(B)|$, minimizing $\max\{|f(R)|,|f(B)|\}$ and minimizing $\max\{|\lambda(R)|, |\lambda(B)|\}$.  Here, $|\cdot|$ denotes the cost (e.g., sum of edge lengths) of the structure. 

The problems we study are natural variants of well-studied network
optimization problems.  Our motivation comes also from a model of
secure connectivity in networks involving facilities with replicated
data.  Consider a set of facilities each having two (or more)
replications of their data; the facilities are associated with pairs
of points (or $k$-tuples of points in the case of higher levels of
replication).  Our goal may be to compute two networks (a ``red''
network and a ``blue'' network) to interconnect the facilities, each
network visiting exactly one data site from each facility; for
communication connectivity, we would require each network to be a
tree, while for servicing facilities with a mobile agent, we would
require each network to be a Hamiltonian path/cycle.  By keeping the
red and blue networks distinct, a malicious attack within one network
is isolated from the other.
\\ \\ \\
\emph{Our results.}

We show that several of these problems are NP-hard and give O(1)-approximation algorithms for each of them. Table~\ref{results_table} summarizes our O(1)-approximation results.
\begin{table}[h!] 
        \begin{tabular}{| c | c | c | c |}
        \hline
        	&	min $|f(R)| + |f(B)|$	&	min-max$\{|\lambda(R)|, |\lambda(B)|\}$	&	min-max$\{|f(R)|,|f(B)|\}$\\ \hline
               
        Spanning tree		& $3\alpha$& \begin{tabular}{cc} \multirow{2}{*} & \hspace{-10mm} 9 \\ & \hspace{-5mm} 3  for $\mathbb{R}$ \end{tabular}	& $4\alpha$  \\ \hline
        
        Matching & 2	& 	\hspace{-6mm} 3 & 3\\ \hline

        TSP tour		& $3\beta$	& \hspace{-5mm} 18	& $6\beta$ \\ \hline
        \end{tabular}
        
\caption{Table of results: $\alpha$ is the Steiner ratio and $\beta$ the best approximation factor
        of the TSP in the underlying metric space. Unless specified otherwise, all other results in this
         table apply to general metric spaces.}
             \label{results_table}
\end{table}

\noindent \emph{Related work.}

Several optimization problems have been studied of the following sort:
Given sets of tuples of points (in a Euclidean space or a general
metric space), select exactly one point or at least one point from
each tuple in order to optimize a specified objective function on the
selected set. Gabow et al.~\cite{gabow1976} explored the problem in
which one is given a directed acyclic graph with a source node $s$ and
a terminal node $t$ and a set of $k$ pairs of nodes, where the objective
was to determine if there exists a path from $s$ to $t$ that uses at
most one node from each pair.  Myung et al.~\cite{myung1995}
introduced the Generalized Minimum Spanning Tree Problem: Given an
undirected graph with the nodes partitioned into subsets, compute a
minimum spanning tree that uses exactly one point from each
subset. They show that this problem is NP-hard and that no
constant-factor approximation algorithm exists for this problem unless
$P=NP$.  Related work addresses the generalized traveling salesperson
problem~\cite{Binay, Pop2004,Pop2001,  slavik}, in which a tour must
visit one point from each of the given subsets.  Arkin et
al.~\cite{estie2000} studied the problem in which one is given a set
$V$ and a set of subsets of $V$, and one wants to select at least one
element from each subset in order to minimize the diameter of the
chosen set. They also considered maximizing the minimum distance
between any two elements of the chosen set.  In another recent paper,
Consuegra et al.~\cite{consuegra} consider several problems of this
kind. Abellanas et al.~\cite{abellanas}, Das et al.~\cite{das}, and
Khantemouri et al.~\cite{Khanteimouri} considered the following
problem. Given colored points in the Euclidean plane, find the
smallest region of a certain type (e.g., strip, axis-parallel square,
etc.) that encloses at least one point from each color.  Barba et
al.~\cite{barba} studied the problem in which one is given a set of
colored points (of $t$ different colors) in the Euclidean plane and a
vector $c = (c_1, c_2, \dots, c_t)$, and the goal is to find a region
(axis-aligned rectangle, square, disk) that encloses exactly $c_i$
points of color $i$ for each $i$.  Efficient algorithms are given for
deciding whether or not such a region exists for a given $c$.

While optimization problems of the ``one of a set'' flavor have been
studied extensively, the problems we study here are fundamentally
different: we care not just about a single structure (e.g., network)
that makes the best ``one of a set'' choices on, say, pairs of points;
we must consider also the cost of a second network on the ``leftover''
points (one from each pair) {\em not} chosen.  As far as we know, the
problem of partitioning points from pairs into two sets in order to
optimize objective functions on \emph{both} sets has not been
extensively studied. One recent work of Arkin et al.~\cite{arkin15}
does address optimizing objectives on both sets: Given a set of pairs
of points in the Euclidean plane, color the points red and blue so
that if one point of a pair is colored red (resp. blue), the other
must be colored blue (resp. red). The objective is to optimize the
radii of the minimum enclosing disk of the red points and the minimum
enclosing disk of the blue points. They studied the objectives of
minimizing the sum of the two radii and minimizing the maximum radius.

\section{Spanning Trees}
Let $MST(X)$ be a minimum spanning tree over the point set $X$, and  $|MST(X)|$ be the cost of the tree, i.e. sum of edge lengths.  Let $\lambda(X)$ be the longest edge in a bottleneck spanning tree on point set $X$ and $|\lambda(X)|$ be the cost of that edge. Given $n$ pairs of points in a metric space, find a feasible coloring which minimizes the cost of a pair of spanning trees, one built over each color class.

\subsection{Minimum Sum}
 In this section we consider minimizing $|MST(R)| + |MST(B)|$.

\begin{theorem} 
The Min-Sum 2-MST problem is NP-hard in general metric spaces. \normalfont{[The proof is in the appendix.]}
\label{thm:minsum_2mst_hard}
\end{theorem}

\noindent \textbf{An $O(1)$-approximation algorithm for Min-Sum 2-MST problem.} \\
Compute $MST(S)$, a minimum spanning tree on all $2n$ points.
Imagine removing the heaviest edge, $h$, from $MST(S)$. This leaves us with two  trees; $T_1$ and $T_2$.
Perform a preorder traversal on $T_1$, coloring nodes red as long as there is no conflict. If there is a conflict ($q_i$ is reached in the traversal and $p_i$ was already colored to red) then color the node blue. Repeat this for $T_2$. We then return the coloring $S = R \cup B$ as our approximate coloring. 

\begin{itemize}
\item Case 1: All nodes in $T_1$ are of the same color and all nodes in $T_2$ are of the same color. \\

This partition is optimal.  To see this, note that the weight of $MST(S)\setminus \{h\}$ is a lower bound on the cost of the optimal solution as it is the cheapest way to create two trees, the union of which span all of the input nodes. Since each tree is single colored, we know that each tree must have $n$ points, exactly one from each pair, and thus is also feasible to our problem.

\item Case 2: One tree is multicolored and the other is not. \\

Let $OPT$ be the optimal solution.  Suppose without loss of generality that $T_1$ contains only red nodes and $T_2$ contains both blue and red nodes.  Then,   there must be a pair with both nodes in $T_2$.  Imagine also constructing an MST on each color class of an optimal coloring.  By definition, in the MSTs built over each color class, at least one point in $T_2$ must be connected to a point in $T_1$.  This implies that the weight of the optimal solution is at least as large as $|h|$, as $h$ is the cheapest edge which spans the cut $(T_1, T_2)$.  Therefore, $|h| \leq |OPT|$.

Consider $MST(R)$. By the Steiner property, we have that an MST over a subset $U \subseteq S$ has weight at most $\alpha|MST(S)|$ where $\alpha$ is the Steiner ratio of the metric space.  Recall that $|MST(S)\setminus \{h\} | \leq |OPT|$. In this case, since $|h| \leq |OPT|$, we have that $|MST(R)| \le \alpha |MST(S)| \le  2 \alpha |OPT|$.  

Next, consider building $MST(B)$. Since no blue node exists in $T_1$, there does not exist an edge that crosses the cut $(T_1, T_2)$ in $MST(B)$, and thus we have that $|MST(B)| \leq \alpha |MST(S)\setminus \{h\} | \leq \alpha |OPT|$. Therefore, $|MST(R) \cup MST(B)|\leq 3\alpha|OPT|$.

\item Case 3: Both trees are multicolored. \\

In this case, there are two pairs one with both nodes contained in $T_1$ and one with both nodes contained in $T_2$.  Imagine, again, constructing an MST on each color class in this optimal coloring. In this case, there must be at least two edges crossing the cut $(T_1, T_2)$, one edge belonging to each tree. Note that each of these edges has weight at least $|h|$ as $h$ is the cheapest edge spanning the cut $(T_1, T_2)$, implying that $|h| \leq |OPT|/2$. Thus, $|MST(S)| \le 1.5|OPT|$ as $|MST(S)\setminus \{h\} | \le |OPT|$ and $|h| \leq |OPT|/2$.

Using our approximate coloring, one can compute $MST(B)$ and $MST(R)$, each with weight at most $\alpha |MST(S)|$. Therefore $|MST(R) \cup MST(B)| \le 2\alpha|MST(S)| \le 3\alpha |OPT|$, where $\alpha$ is again the Steiner ratio of the metric space.

 \end{itemize}
 
\noindent Using the above case analysis, we have the following theorem.

\begin{theorem}
There exists a $3\alpha$-approximation for the Min-Sum 2-MST problem.
\end{theorem}
\textbf{Remark:}The Steiner ratio is the supremum of the ratio of length of an minimum spanning tree and a minimum Steiner tree over a point set.  In a general metric space $\alpha = 2$ and in the Euclidean plane $\alpha \leq 1.3546$~\cite{Ismailescu}.

\subsection{Min-max}
In this section the objective is to $\min \max \{|MST(R)|, |MST(B)|$\}.

\begin{theorem}
The Min-Max 2-MST problem is strongly NP-hard in general metric spaces. 
\label{thm:minmax_2mst_strong_hard}
\end{theorem}
\begin{proof}
The reduction is from a problem which we will call {\em connected partition} \cite{dyer1985complexity}.  In {\em connected partition} one is given a graph $G = (V, E)$, where $|V| = n$, and asked if it is possible to remove a set of edges from $G$ which breaks it into two connected components each of size $n/2$.

Given an instance of {\em connected partition}, $G = (V, E)$, we will create an instance of min-max 2-MST as follows.  For each vertex $v_i \in V$ create an input pair $\{p_i, q_i\}$.  For each edge $e = (v_i, v_j) \in E$ set the distance between the corresponding points,  $p_i$ and $p_j$ to be one. Set the distances $d(q_i, q_j)$ to be zero for all $i, j$, and the distances $d(p_i, q_j)$ to be two for all $i, j$.  In order to complete the construction, set all remaining distances to be the shortest path length among the distances defined above.

{\bf Claim:} $G$ can be partitioned into two connected components of size $n/2$ if and only if there is a solution to the corresponding instance of min-max 2-MST with value $n/2 + 1$.

To show the first direction, suppose that the graph $G$ can be split into two connected components, $C_1, C_2$, of size exactly $n/2$.  Without loss of generality suppose $\{v_i \: : 1 \leq i \leq n/2\} \in C_1$ and $\{v_j \: : n/2 < j \leq n\} \in C_2$.  Then, it is easy to verify based on the pairwise distances in the metric space described above that the coloring $\{p_i \: : 1 \leq i \leq n/2\} \cup \{q_j \: : n/2 < j \leq n\} \in R$, $\{p_i \: : n/2 < i \leq n\} \cup \{q_j \: : 1 \leq j \leq n/2\} \in B$, achieves a cost of $n/2 + 1$.

To show the opposite direction, suppose that there is a solution to the instance of min-max 2-MST of cost $n/2 + 1$.  Notice that the minimum distance from point $p_i$ to any other point is at least one; therefore, there can be at most $n/2 + 2$ points from the set $P = \{p_i \: : 1 \leq i \leq n\}$ colored either red or blue in the solution which achieves this cost.  Thus there are at least $n/2 - 2$ points from the set $Q = \{q_j \: : 1 \leq j \leq n\}$ colored either red or blue in this solution in order for it to be a feasible coloring.  This implies that there will be at least one edge crossing the cut $(P, Q)$ in both the red and blue MST which realize the cost of this solution, and this edge has cost two.  Then, of the remaining budget of $n/2 - 1$ units in order to complete the trees which realize the cost of this solution, it must be the case that we can utilize $n/2 - 1$ edges of length one which interconnect exactly $n/2$ nodes from the set $P$ in each color class.

The edges of length one in our metric space correspond directly to original edges of the graph $G$ in {\em connected partition} thus showing that there exists two spanning trees each of which spans exactly $n/2$ nodes of $G$ and thus $G$ can be partitioned into two connected components of size exactly $n/2$.
\end{proof}

\begin{theorem}
There exists a 4$\alpha$-approximation for the Min-Max 2-MST problem.
\label{thm:minmax_2mst}
\end{theorem}
\begin{proof}
We use the same algorithm as we did for the Min-Sum 2-MST problem. The approximation factor is dominated by case 2 in the Min-Sum 2-MST analysis. For the Min-Max objective function, we have that $\max\{|MST(B)|,$ $ |MST(R)|\} \leq \alpha |MST(S)|$ and that $|MST(S)| \leq 4|OPT|$. Thus, $\max\{|MST(B)|, $ $|MST(R)|\} \leq 4\alpha |OPT|$.
\end{proof}

\subsection{Bottleneck}
\label{sec:bmst}
In this section the objective is to $\min \max\{|\lambda(R)|, |\lambda(B)|\}$.

\begin{lemma}
Given $n$ pairs of points on a line in $\mathbb{R}^2$ where consecutive points on the line are unit separated, there exists a feasible coloring of the points, such that $\max\{|\lambda(R)|, |\lambda(B)|\} \leq 3$.
\label{lemma:gapbound}
\end{lemma}

\begin{proof}
The proof will be constructive, using Algorithm \ref{algorithm:color}. We partition the points into $n$ disjoint buckets, where a bucket consists of two consecutive points on the line.

\begin{algorithm}[h!]
Color the leftmost point, $p$, red\\
Let $p'$ be the point that is in $p$'s bucket \\
Let $R$ be a set of red points and $B$ be a set of blue points \\
$R \leftarrow \{p\}$; $B \leftarrow \emptyset$ \\
\While{There exists an uncolored point}{
\While{$p'$ is uncolored}{
\eIf{$p$ is red}{
Color $p$'s pair, $q$, blue \\
$B \leftarrow B \cup \{q\}$ \\
$p \leftarrow q$ \\
}{
Let $p''$ be the point in $p$'s bucket \\
Color $p''$ red \\
$R \leftarrow R \cup \{p''\}$ \\
$p \leftarrow p''$ \\
}
}
Find the leftmost uncolored point $x$ and color it red. Let $x'$ be the point in $x$'s bucket \\
$p \leftarrow x$; $p' \leftarrow x'$ \\
}

return $\{R, B\}$
\caption{Coloring points on a line.}
\label{algorithm:color}
\end{algorithm}

Observe that at the end of Algorithm~\ref{algorithm:color}, each bucket has exactly one red point and one blue point. Thus, the maximum distance between any two points of the same color is 3.
\end{proof}

\begin{theorem}
There exists a 3-approximation algorithm for the Bottleneck 2-MST problem on a line.
\label{theorem:pathonedim}
\end{theorem}

\begin{proof}
Note that if the leftmost $n$ points do not contain two points from the same pair, then it is optimal to let $R$ be the leftmost $n$ points and $B$ be the rightmost $n$ points. Suppose now that the leftmost $n$ points contain two points from the same pair. We run Algorithm~\ref{algorithm:color} on the input.  Imagine building two bottleneck spanning trees over the approximate coloring as well as over an optimal coloring. Let $\lambda$ be the longest edge (between two points of the same color) in our solution and $\lambda^*$ be the longest edge in the optimal solution.

Consider any two consecutive input points $s_i$ and $s_{i+1}$ on the line. We first show that $|\lambda^*| \geq |s_is_{i+1}|$ by arguing that the optimal solution must have an edge that covers the interval $[s_i,s_{i+1}]$. Suppose to the contrary that no such edge exists. This means that $s_i$ is connected to $n-1$ points only to its left and $s_{i+1}$ is connected to $n-1$ points only to its right. This contradicts the assumption that the leftmost $n$ points contain two points from the same pair.

Let the longest edge in our solution be defined by two points, $p_i$ and $p_j$. Consider the number of input points in interval $[p_i, p_j]$. Input points in this interval other than $p_i$ and $p_j$ will have a different color than $p_i$ and $p_j$. It is easy to see that if $[p_i, p_j]$ consists of two input points, that $|\lambda^*| = |\lambda|$, and if $[p_i, p_j]$ consists of three input points, that $|\lambda^*| \geq |\lambda/2|$. We know by lemma~\ref{lemma:gapbound} that $[p_i, p_j]$ can consist of no more than four input points. In this last case, $|\lambda^*|$ must be at least the length of the longest edge of the three edges in $[p_i, p_j]$. Thus, we see that $|\lambda^*| \geq |\lambda|/3$.
\end{proof}

\begin{theorem}
There exists a 9-approximation algorithm for the Bottleneck 2-MST problem in a metric space.
\label{theorem:pathmetric}
\end{theorem}

\begin{proof}
First, we compute $MST(S)$ and consider the heaviest edge, $h$. The removal of this edge separates the nodes into two connected components, $H_1$ and $H_2$. If $\nexists \: i \: : p_i, q_i \in H_j$ for $1 \leq i \leq n$ and $1 \leq j \leq 2$, then we let $R = H_1$ and $B = H_2$ and return $R$ and $B$.  Let $\lambda^*$ be the heaviest edge in the bottleneck spanning trees built on an optimal coloring. Note that $MST(S)$ lexicographically minimizes the weight of the $k$th heaviest edge, $1 \le k \le 2n - 1$,  among all spanning trees over $S$, and thus the weight of the heaviest edge in $MST(S) \setminus \{h\}$ is a lower bound on $|\lambda^*|$. Thus, in this case, our solution is clearly feasible and is also optimal as $MST(R)$ and $MST(B)$ are subsets of $MST(S)\setminus \{h\}$.

Now suppose $\exists \: j \: \in \{1,2\}: p_i, q_i \in H_j, 1\le i \le n$. This means that $|\lambda^*| \geq |h|$.  In this case, we compute a bottleneck TSP tour on the entire point set. It is known that that a bottleneck TSP tour with bottleneck edge $\lambda$ can be computed from $MST(S)$ so that $|\lambda| \leq 3|h| \leq 3|\lambda^*|$.

Next we run Algorithm~\ref{algorithm:color} on the TSP tour and return two paths, each having the property that the largest edge has weight no larger than $9|\lambda^*|$. 
\end{proof}

\textbf{Remark:} Consider the problem of computing a feasible partition which minimizes the bottleneck edge across two bottleneck TSP tours. Let the heaviest edge in the bottleneck TSP tours built on the optimal partition be $\lambda^{**}$. The above algorithm gives a 9-approximation to this problem as well because the algorithm returns two Hamilton paths and we know that (using the notation in the above proof) $|\lambda^{*}| \leq |\lambda^{**}|$. Thus, $|\lambda| \leq 9|\lambda^*| \leq 9|\lambda^{**}|$.\\

The following is a generalization of Lemma~\ref{lemma:gapbound}.
Let $\mathcal{S}=\{S_1, S_2, \dots ,S_n\}$ be a set of $n$ $k$-tuples of points on a line. 
Each set $S_i$, $1 \leq i \leq n$, must be colored with $k$ colors. That is, no two points in set $S_i$ can be of the same color.

Consider two consecutive points of the same color, $p$ and  $q$. 
We show that there exists a polynomial time algorithm that colors the points in $S$ so that the number of input points in interval $(p, q)$ is $O(k)$.

\begin{lemma}
There exists a polynomial time algorithm to color $S$ so that for any two consecutive input points of the same color, $p$ and $q$, the interval $(p, q)$ contains at most $2k - 2$ input points. \label{lem:k}
\end{lemma}
\begin{proof}
The algorithm consists of $k$ steps, where in the $j$th step, we color $n$ of the yet uncolored points with color $j$.  We describe the first step.
 
Divide the $kn$ points into $n$ disjoint buckets, each of size $k$, where the first bucket $B_1$ consists of the $k$ leftmost points, the second bucket $B_2$ consists of the points in places $k+1,k+2,\dots,2k$, etc.
Let $G = (V, E)$ be the bipartite graph, with node set $V = \{\mathcal{S} \cup B=\{B_1,\ldots,B_n\}\}$, in which there is an edge
between $B_i$ and $S_j$ if and only if at least one of $S_j$'s points lies in bucket $B_i$. According to Hall's theorem~\cite{Hall}, there exists a perfect matching in $G$. Let $M$ be such a matching and for each edge $e=(B_i,S_j)$ in $M$, color one of the points in $B_i \cap S_j$ with color 1.  Now, remove from each tuple the point that was colored 1, and remove from each bucket the point that was colored 1.  
In the second step we color a single point in each bucket with the color 2, by again computing a perfect matching between the buckets (now of size $k-1$) and the $(k-1)$-tuples.  It is now easy to see that for any two consecutive points of the same color, $p$ and $q$, at most $2k - 2$ points exist in interval $(p, q)$.
\end{proof}

\section{Matchings}
Let $M(X)$ be the minimum weight matching on point set $X$ and $|M(X)|$ be the cost of the matching. Let $\lambda(X)$ be the longest edge in a bottleneck matching on point set $X$ and $|\lambda(X)|$ be the cost of that edge edge.  Given $n$ pairs of points in a metric space, find a feasible coloring which minimizes the cost of a pair of matchings, one built over each color class.
\vspace{-3mm}
\subsection{Minimum Sum}

In this section the objective is to minimize $|M(R)| + |M(B)|$.

\begin{theorem}
There exists a 2-approximation for the Min-Sum 2-Matching problem in general metric spaces.
\end{theorem}

\begin{proof}
First, note that the weight of the minimum weight perfect matching on $S$, $M^*$, which forbids edges $(p_i, q_i)$ for all $i$ is a lower bound on $|OPT|$. Next, we define the minimum weight one of a pair matching, $\hat{M}$, to be a minimum weight perfect matching which uses exactly one point from each input pair $\{p_i, q_i\}$ (that is, a matching using the most advantageous point from each input pair to minimize the total weight of the matching.)   Observe that $|\hat{M}|$, is a lower bound on the weight of the smaller of the matchings of OPT and therefore has weight at most $|OPT|/2$.

Our algorithm is to compute $\hat{M}$, and color the points involved in this matching red, and the remainder blue.  We return the coloring $R \cup B$ as our approximate solution.

We have that $|M(R)| =  |\hat{M}| \le |OPT|/2$. To bound $|M(B)|$, consider the multigraph $G = (V = S, E = M^* \cup \hat{M})$.  All $v \in B$ have degree 1 (from $M^*$), and all $u \in R$ have degree 2 (from $M^*$ and $\hat{M}$).  For each $v_i \in B$, either $v_i$ is matched to $v_j \in B$ by $M^*$, or  $v_i$ is matched to $u_i \in R$ by $M^*$.  In the former case we can consider $v_i$ and  $v_j$ matched in $B$  and charge the weight of this edge to $|M^*|$.  In the latter case, note that each $u \in R$ is part of a unique cycle, or a unique path.  If $u \in R$ is part of a cycle then no vertex in that cycle belongs to $B$ due to the degree constraint.  Thus, if $v_i \in B$ is matched to $u_i \in R$, $u_i$ is part of a unique path whose other terminal vertex $x$ belongs to $B$, due to the degree constraint.  We can consider $v_i$, and  $x$ matched and charge the weight of this edge to the unique path connecting $v_i$ and  $x$ in $G$.  Thus, $|M(B)|$ can be charged to $|M^* \cup \hat{M}|$ and has weight at most $1.5|OPT|$.

Therefore, our partition guarantees $|M(R)| + |M(B)| \le 2|OPT|$.  Figure \ref{fig:2match_sum_tight} shows the approximation factor using our algorithm is tight.
\end{proof}

\begin{figure}[t]
\centering
\includegraphics[scale = 0.5]{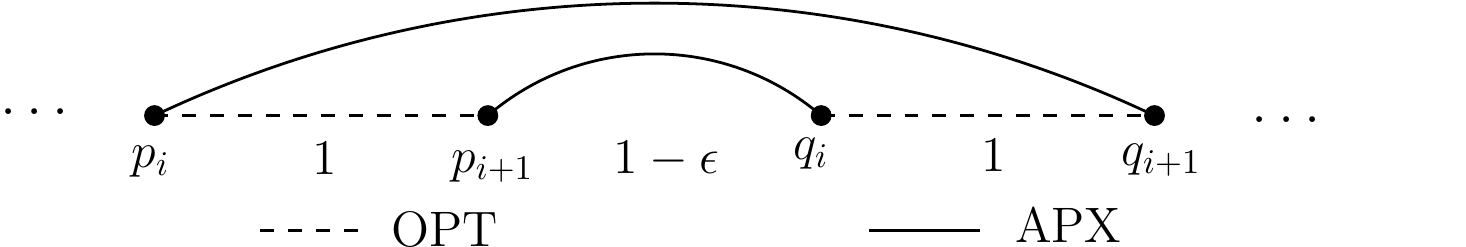}
\caption{$\frac{|APX|}{|OPT|} \approx 2$ }
\label{fig:2match_sum_tight}
\end{figure}

\vspace{-3mm}

\subsection{Min-max}
In this section the objective is to $\min \max\{|M(R)|,|M(B)|\}$.

\begin{theorem}
The Min-Max 2-Matching problem is weakly NP-hard in the Euclidean plane.\label{thm:minmax_m_hard}
\end{theorem}

\begin{proof}
The reduction is from {\sc Partition}: given a set $P=\{x_1, x_2, ..., x_n\}$ of $n$ integers, decide if there exists a partition $P = P_1 \cup P_2$, with $\sum_{i \in P_1}x_i = \sum_{j \in P_2}x_j$.  Let $M = \sum_i x_i$.  Given any instance $P$ of {\sc Partition}, we create a geometric instance of the Min-Max  2-Matching problem, as shown in Figure \ref{fig:minmaxM_hard}.

We place $n$ point pairs $\{p_i, q_i\}_{i = 1}^n$ along two $\epsilon$-separated horizontal lines, such that $p_i, q_i$ are vertically adjacent, with horizontal separation of $M$ between consecutive pairs. Then, for each $x_i$ in the instance of {\sc Partition} we  place a point $p_{n + i}$ at distance $x_i$ from $p_i$, and its corresponding pair $q_{n + i}$ at distance $\epsilon/2$ from both $q_i$ and $p_i$.

Notice that any solution which minimizes the weight of the larger matching created only uses edges  between points of the same ``cluster'' $\{p_i , q_i, p_{n + i}, q_{n + i}\} $.  Any edge between two clusters $\{p_i , q_i, p_{n + i}, q_{n + i}\}$, $\{p_j , q_j, p_{n + j}, q_{n + j}\}, i \ne j$ costs at least $M$ and if matching edges are chosen within clusters the entire matching can be constructed with cost at most $M + o(1)$ for $\epsilon > 0$ chosen small enough.

Within each cluster an assignment will have to be made, that is, without loss of generality, $\{p_i, p_{n + i}\} \in R, \{q_i, q_{n + i}\} \in B$ or $\{p_i, q_{n + i}\} \in R, \{q_i, p_{n + i}\} \in B$.  Therefore,  any algorithm that minimizes the maximum weight of either matching also minimizes $\max\{\sum_{i\in P_1}x_i,  \sum_{j\in P_2} x_j\}$ across all partitions $P_1 \cup P_2$.  Thus, for $\epsilon > 0$ chosen small enough the instance of partition is solvable if and only if the weight of the larger matching created is at most $\frac{M}{2} + o(1)$.

\begin{figure}[t]
\centering
\includegraphics[scale=0.7]{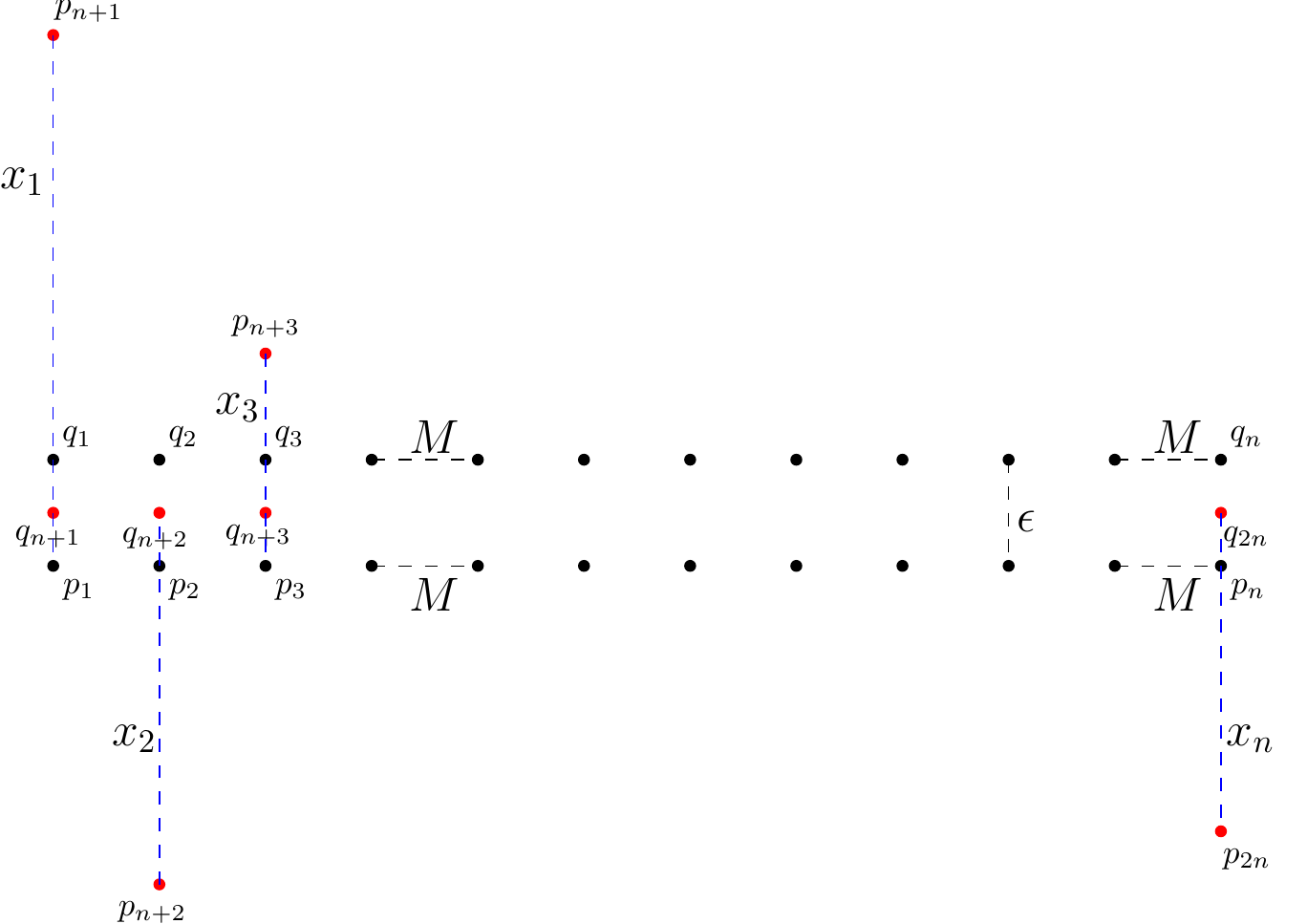}
\caption{Set up of the Min-Max 2-Matching instance given an instance of {\sc Partition}: $\{x_1, x_2, ..., x_n\}$.}
\label{fig:minmaxM_hard}
\end{figure}
\end{proof}

\textbf{Remark:} The above reduction can also be used to show that the {\em Min-Max 2-MST problem} is weakly NP-hard in the Euclidean plane.  Given an instance of partition we create the exact same instance for {\em Min-Max 2-MST} as described above, and note that there exists a solution for partition if and only if there is a solution to {\em Min-Max 2-MST} with value $(n - 1)M + M/2 + o(1)$.
\\

\vspace{-5mm}

\begin{theorem}
The approximation algorithm for the Min-Sum 2-Matching problem serves as a 3-approximation for the Min-Max 2-Matching problem in  general metric spaces. \normalfont{[The proof is in the appendix.]}
\label{thm:minmax_m_apx}
\end{theorem}

\vspace{-5mm}
\subsection{Bottleneck}
In this section the objective is to $\min \max\{|\lambda(R)|, |\lambda(B)|\}$.

\begin{theorem}
There exists a 3-approximation to the Bottleneck 2-Matching problem in general metric spaces. \normalfont{[The proof is in the appendix]}
\label{thm:bottleneck_m}
\end{theorem}

\section{TSP Tours}
Let $TSP(X)$ be a TSP tour on point set $X$ and $|TSP(X)|$ be the cost of the tour. Let $\lambda(X)$ be the longest edge in a bottleneck TSP tour on point set $X$ and $|\lambda(X)|$ be the cost of that longest edge. Given $n$ pairs of points in a metric space, find a feasible coloring which minimizes the cost of a pair of TSP tours, one built over each color class.

It is interesting to note the complexity difference emerging here.  In prior sections, the structures to be computed on each color class of a feasible coloring were computable exactly in polynomial time.  Thus, the decision versions of these problems, which ask if there exists a feasible coloring such that some cost function over the pair of structures is at most $k$,  are easily seen to be in NP.  However, when the cost function is over a set of TSP tours or bottleneck TSP tours, this is no longer the case.  That is, suppose that a non-deterministic Touring machine could in polynomial time, for a point set $S$ and $k \in \mathbb{R}$, return a coloring for which it claimed the cost of the TSP tours generated over both color classes is at most $k$.  Unless $P=NP$, the verifier cannot in polynomial time confirm that this is a valid solution, and therefore the problem is not in NP.  Thus, the problems considered in this section are all NP-hard.

\vspace{-3mm}
\subsection{Minimum Sum}
\label{sec:tsp_min_sum}
In this section the objective is to minimize $|TSP(R)| + |TSP(B)|$.

\begin{algorithm}[h!]
Let $\overline{TSP_\beta}(X)$ denote a $\beta$-factor approximate TSP tour on set $X$.
\begin{enumerate}
\item Compute $\overline{TSP_\beta}(S)$.
\item Let $2k$ be the largest even number not exceeding $(2+\frac{1}{\mu}) \beta$.
Enumerate all ways of decomposing $\overline{TSP_\beta}(S)$ into $2k$ connected components:
for each decomposition, color the nodes from consecutive components red and blue alternately
(i.e. color all nodes in component one red, all nodes in component two
blue, etc.). If this coloring is infeasible, then skip to the next decomposition;
otherwise compute $\overline{TSP_\beta}(R)$ and $\overline{TSP_\beta}(B).$
\item Compute a random feasible coloring, ${S} = R \cup  B$, and compute $\overline{TSP_\beta}(R)$ and $\overline{TSP_\beta}(B).$
\item Among all pairs of tours produced in steps 2 and 3, choose the pair of minimum sum.
\end{enumerate}

\caption{Algorithm $A(\mu, \beta)$. $0 < \mu < 1$ and $\beta > 1$.}
\label{alg:tsp}
\end{algorithm}

We will show for $\beta > 1$ and for the proper choice of $\mu$, that Algorithm~\ref{alg:tsp} gives a $3\beta$-approximation for the Min-Sum 2-TSP problem.  Fix a constant $\mu<1$. Let $OPT$ be the optimal (feasible) coloring $S =R^* \cup B^*$.  Let $d(R, B)$ be the minimum point-wise distance between sets $R$ and $B$.  We call an instance of the problem $\mu$-separable if there exists a feasible coloring ${S} = R \cup B : d(R, B)\geq \mu(|TSP(R)|+|TSP(B)|).$

Let $APX$ be the coloring returned by our algorithm. We will show that if ${S}$ is not $\mu$-separable, then $|APX|\leq\frac{2}{1-4\mu}\beta|OPT|$ (see Lemma~\ref{lem:not_mu_sep}) and that if ${S}$ is $\mu$-separable, then $|APX|\leq\frac{1}{4\mu}\beta|OPT|$ (see Lemma~\ref{lem:mu_sep}). Supposing both of these are true, then the approximation factor of our algorithm is $\max\{\frac{1}{4\mu},\frac{2}{1-4\mu}\}  \beta$.  One can easily verify that $\mu=1/12$ is the minimizer which gives the desired $3\beta$ factor.  The following lemma states that if ${S}$ is not $\mu$-separable, then any feasible coloring yields a ``good'' approximation.

\begin{lemma}
If ${S}$ is not $\mu$-separable, then $|APX|\leq\frac{2}{1-4\mu}\beta|OPT|$.
\label{lem:not_mu_sep}
\end{lemma}
\begin{proof}
If ${S}$ is not $\mu$-separable, then for any feasible coloring $S=R \cup B$ we have $d(R,B)\leq\mu(|TSP(R)|+|TSP(B)|)$.  In particular, for the coloring induced by the optimal solution, $S=R^{*} \cup B^{*}$,  $d(R^{*},B^{*})\leq\mu(|TSP(R^{*})|+|TSP(B^{*})|).$  Then,
\[|TSP({S})|\leq |OPT|+2d(R^{*},B^{*}) \leq |OPT|+2\mu(|TSP(R^{*})|+|TSP(B^{*})|) \]
\[\leq |OPT|+4\mu |TSP({S})|.\]
Hence, when $\mu<\frac{1}{4}$,  $|TSP(S)|\leq\frac{1}{1-4\mu}|OPT|$.  Let ${S}=\hat{R} \cup \hat{B}$ be the random feasible coloring computed by $\mathcal{A}(\mu,\beta)$. Then, as we are returning the best coloring between $\hat{R} \cup \hat{B}$ and all $O(n^{2k})$  colorings of $\overline{TSP_\beta}(S)$,  we have $|APX|\leq \beta(|TSP(\hat{R})|+|TSP(\hat{B})|)\leq2\beta|TSP(\mathcal{S})|\leq\frac{2\beta}{1-4\mu}|OPT|$.
\end{proof}
The following lemma states that if ${S}$ is $\mu$-separable, then any witness coloring to the $\mu$-separability of ${S}$ gives a ``good'' approximation.

\begin{lemma}
If ${S}$ is $\mu$-separable, then $|APX|\leq\frac{1}{4\mu}\beta|OPT|$.
\label{lem:mu_sep}
\end{lemma}

\begin{proof}
Suppose we successfully guessed a coloring $X_0 = R^0 \cup B^0$ that is a ``witness'' to the $\mu$-separability of ${S}$ (we will show how to guess $X_0$ later).

\begin{itemize}

\item Case 1: $OPT=X_{0}.$ Then $|APX|\leq \beta(|TSP(R^{0})|+|TSP(B^{0})|)
=\beta(|TSP(R^{*})|+|TSP(B^{*})|)= \beta|OPT|$.

\item Case 2: $OPT\neq X_{0}.$ Then $R^* \neq R^0, B^* \neq B^0$
which means each tour in OPT must contain at least 2 edges crossing
the cut ($R^{0}, B^{0}$), hence the optimal solution must contain
at least 4 edges crossing the cut ($R^{0}, B^{0}$). So $|OPT|\geq4d(R^{0},B^{0})\geq4\mu(|TSP(R^{0})|+|TSP(B^{0})|)\geq\frac{4\mu}{\beta}|APX|.$
Equivalently, $|APX|\leq\frac{\beta}{4\mu}|OPT|.$
\end{itemize}
\end{proof}
The next two lemmas show how to guess a witness coloring $X_{0}$ in polynomial time.
First, we show that if ${S}$ is $\mu$-separable with a witness coloring $X_{0}$, then $\overline{TSP_\beta}({S})$ cannot cross the red/blue cut defined by this coloring ``too many'' times.

\begin{lemma}

Let $\overline{TSP_\beta}({S})$ be an $\beta$-factor approximation for $TSP({S}).$ Also, suppose
${S}$ is $\mu$-separable with witness $X_{0}.$ Then $\overline{TSP_\beta}({S})$
crosses the cut $(R^{0}$, $B^{0})$  at most $(2+\frac{1}{\mu}) \beta$
times.
\label{lem:bound_crossing}
\end{lemma}
\begin{proof}
One can construct a TSP tour for ${S}$ by adding two bridges to $TSP(R^{0})$
and $TSP(B^{0})$, thus we have $|TSP({S})|\leq |TSP(R^{0})|+|TSP(B^{0})|+2d(R^{0},B^{0}) \leq(2+\frac{1}{\mu}) d(R^{0},B^{0}).$  Also, suppose $\overline{TSP_\beta}({S})$ crosses the cut ($R^{0}$, $B^{0}$) $2k$ times. Then, $2kd(R^{0},B^{0})\leq$$|\overline{TSP_\beta}({S})|$$\leq \beta |TSP({S})|.$ Combining the above two inequalities, we obtain $2k\leq(2+\frac{1}{\mu}) \beta$.
\end{proof}
The next lemma completes our proof.

\begin{lemma}
Suppose ${S}$ is $\mu$-separable.  Let $X_{0}$ be any coloring which serves as a ``witness''. Then, in step 2 of $\mathcal{A} (\mu,\beta)$, we will encounter $X_{0}$ at some stage.

\end{lemma}

\begin{proof}
Given a nonnegative integer $k$ and a TSP tour $P$, define $\Pi(P,k)=$\{$X$: $X$ is a feasible coloring and $P$ crosses $X$  at most $k$ times\}. By  Lemma~\ref{lem:bound_crossing}, we know $X_{0} \in \Pi(\overline{TSP_\beta}({S}),(2+\frac{1}{\mu}) \beta).$
Since step 2 of $\mathcal{A}(\mu,\beta)$ is actually enumerating all colorings in $\Pi(\overline{TSP_\beta}({S}),(2+\frac{1}{\mu}) \beta),$ this completes the proof.
\end{proof}
Note that step 2 considers $O(n^{2k}) = O(n^{14\beta})$ decompositions and for each coloring that is feasible, we compute two approximate TSP tours. Suppose the running time to compute a $\beta$-factor TSP tour on $n$ points is $h_\beta(n)$. Then the worst case running time of Algorithm~\ref{alg:tsp} is $O(h_\beta(2n)n^{14\beta})$. Thus, we have the following Theorem.

\begin{theorem}

For any $\beta>1$, the algorithm $\mathcal{A}(\frac{1}{12},\beta)$ is a
$3\beta$-approximation for the Min-Sum 2-TSP
problem with running time $O(h_\beta(2n)n^{14\beta})$.
\end{theorem}
\textbf{Remark:} If ${S}$ is in the Euclidean plane then $\beta = 1 + \epsilon$ for some $\epsilon > 0$ \cite{mitchell1999guillotine} yielding a $(3 + \epsilon)$-approximation and if ${S}$ is in a general metric space then $\beta = 3/2$ \cite{christofides1976worst} yielding a 4.5-approximation. In both cases $h_\beta(2n)$ is polynomial.

\vspace{-3mm}
\subsection{Min-Max}
In this section the objective is to $\min \max\{|TSP(R)|, |TSP(B)|\}$.

\begin{theorem}
There exists a $6 \beta$-approximation to the Min-Max 2-TSP problem, where $\beta$ is the approximation factor for TSP in a certain metric space. \normalfont{[The proof is in the appendix.]}
\label{thm:minmax_2tsp}
\end{theorem}

\vspace{-7mm}
\subsection{Bottleneck}
In this section the objective is to $\min \max\{|\lambda(R)|, |\lambda(B)|\}$.
\begin{theorem}
There exists an 18-approximation algorithm for the Bottleneck 2-TSP problem. \normalfont{[The proof is in the appendix.]}
\label{thm:bottleneck_tsp}
\end{theorem}

\newpage

\bibliographystyle{abbrv}
\bibliography{ref}

\newpage

\appendix

\section{Additional Proofs}

\textbf{Theorem~\ref{thm:minsum_2mst_hard}}. \emph{The Min-Sum 2-MST problem is  NP-Hard in general metric spaces.}

\begin{proof}
The reduction is from {\sc Max 2SAT} where one is given $n$ variables $\{x_1, x_2, \dots, x_n\}$ and $m$ clauses $\{c_1, c_2, \dots, c_m\}$. Each clause contains at most two literals joined by a logical \emph{or}. The objective is to maximize the number of clauses that evaluate to \emph{true}.  
	
For each variable $x_i$ we create a variable gadget that consists of two pairs of points:  $\{p_{2i}, q_{2i}\}$ and $\{p_{2i+1}, q_{2i+1}\}$ (see Figure \ref{fig:2-mst-sum-variable-gadget}). Setting $x_i$ to \emph{true} is equivalent to using edges $(p_{2i + 1}, q_{2i})$ and $(p_{2i}, q_{2i + 1})$. Setting $x_i$ to \emph{false} is equivalent to using edges $(p_{2i}, p_{2i + 1})$ and $(q_{2i + 1}, q_{2i})$. Variable gadgets will be arranged on a line with distance $O(L)$ between consecutive variable gadgets for $L = n + m$ (see Figure~\ref{fig:min-sum-2-mst-chain}).
	
For every pair of variable gadgets corresponding to variables $x_i, x_j, i \ne j$  we place a cluster $A_{i,j}$ of $M = m^2$ points near point $p_{2i + 1}$.  Each of these points is paired to a point in a cluster $B_{i,j}$ of $M$ points near point $q_{2j + 1}$. Any two points in the same cluster, $A_{i,j}$ or $B_{i,j}$, are separated by distance two from each other and by distance one from point $p_{2i + 1}, q_{2j + 1}$ respectively. Note that this enforces points $p_{2i + 1}$ and $q_{2j + 1}$ to be in different trees for all $1 \le i,j \le n$.  Otherwise, if $p_{2i + 1}$ and $q_{2j + 1}$ were placed in the same tree, then connecting the points in clusters $A_{i,j}, B_{i,j}$ to the trees would cost at least $M$ more than it would to have $p_{2i + 1}$ and $q_{2j + 1}$ in different trees.  

Now we argue that the optimal solution uses edges $(p_{2i + 1}, p_{2i + 3})$ and $(q_{2i + 1}, q_{2i + 3})$,  $1 \leq i \leq n - 1$, as ``backbones''  of the two MSTs. To see this, observe that if any other edge was used to connect two consecutive variable gadgets, then we would need to use at least one edge of length $L + 2$. Since $p_{2i + 1}$ and $q_{2j + 1}$ will be in different trees for all $1 \leq i,j \leq n$ and since points $p_{2i}$ and $q_{2i}$ will be connected to points $p_{2i + 1}$ and $q_{2i + 1}$ ($1 \leq i \leq n$), we have a set of ``lower'' components that must be connected and a set of ``upper'' components that need to be connected. No upper component can be connected to a lower component. Any edge of length at least $L + 2$ connecting any of these components can thus be replaced by an edge of length $L$. 

	


The remaining variable gadget points, \{$p_{2i}, q_{2i}$\} ($1 \leq i \leq n$), must be connected to the backbones. That is, for variable $x_i$, points $p_{2i}$ and $q_{2i}$ will be picked up either by using edges $(p_{2i + 1}, q_{2i})$ and $(p_{2i}, q_{2i + 1})$ (green in Figure~\ref{fig:truth-assignment}) or edges $(p_{2i}, p_{2i + 1})$ and $(q_{2i + 1}, q_{2i})$ (red in Figure~\ref{fig:truth-assignment}). As mentioned, the green edges correspond to setting $x_i$ to \emph{true} and the red edges correspond to setting $x_i$ to \emph{false}.


	\begin{figure}[h]
	\centering
	\includegraphics[scale = 0.45]{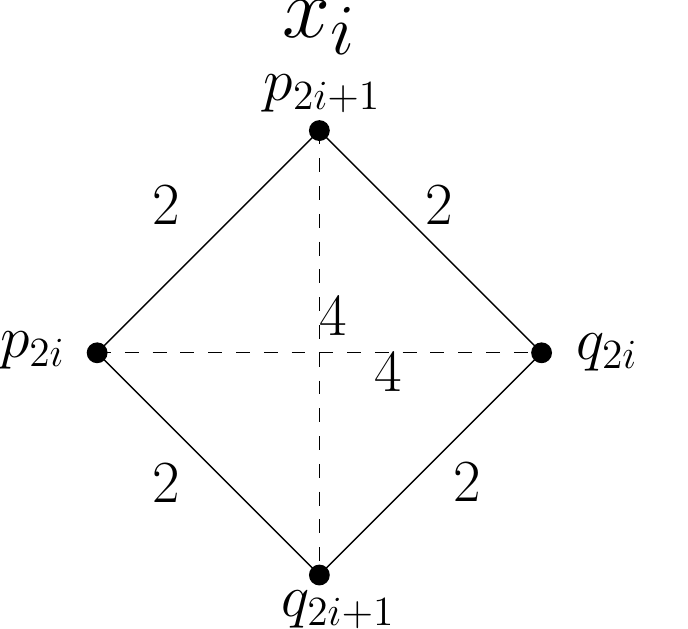}
	\caption{Variable gadget.}
	\label{fig:2-mst-sum-variable-gadget}
	\end{figure}

\begin{figure}[h]
	\centering
	\includegraphics[scale = 0.45]{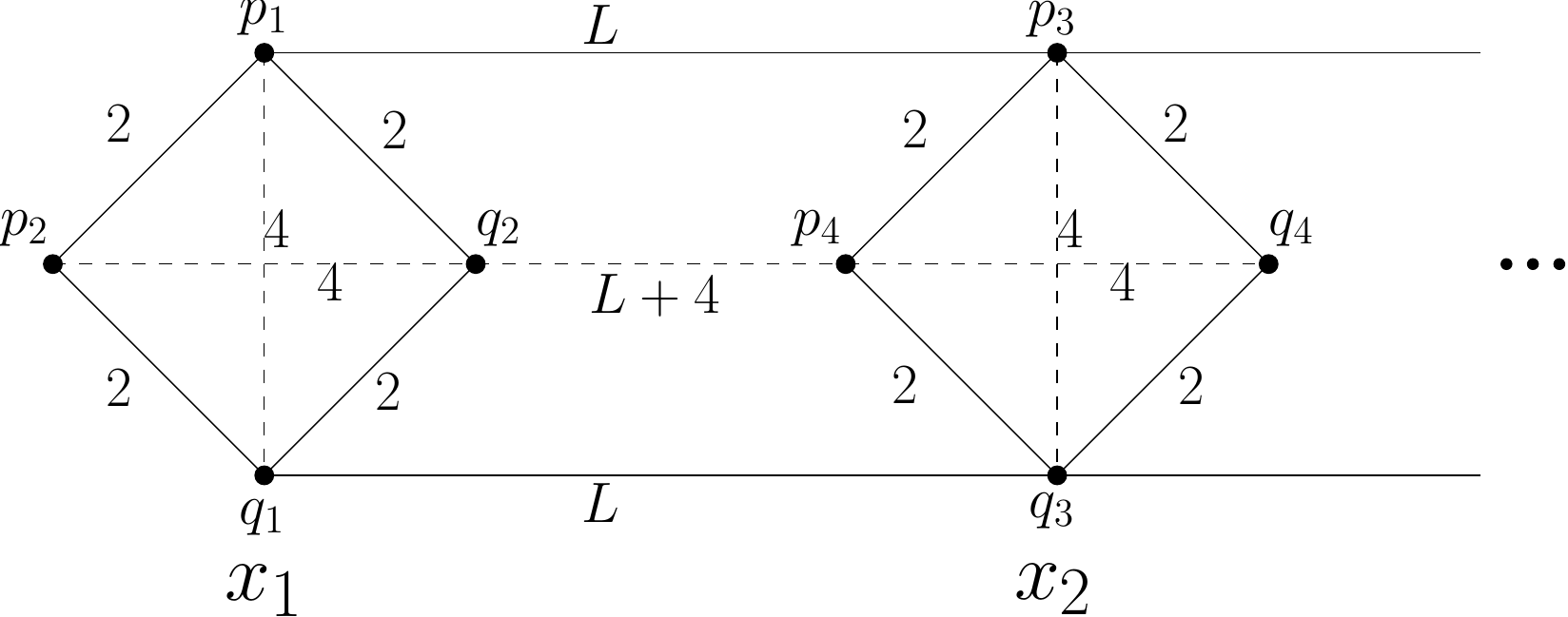}
	\caption{Metric distances between variable gadgets. }
		\label{fig:min-sum-2-mst-chain}
	\end{figure}

\begin{figure}[h]
\centering
\includegraphics[scale = 0.4]{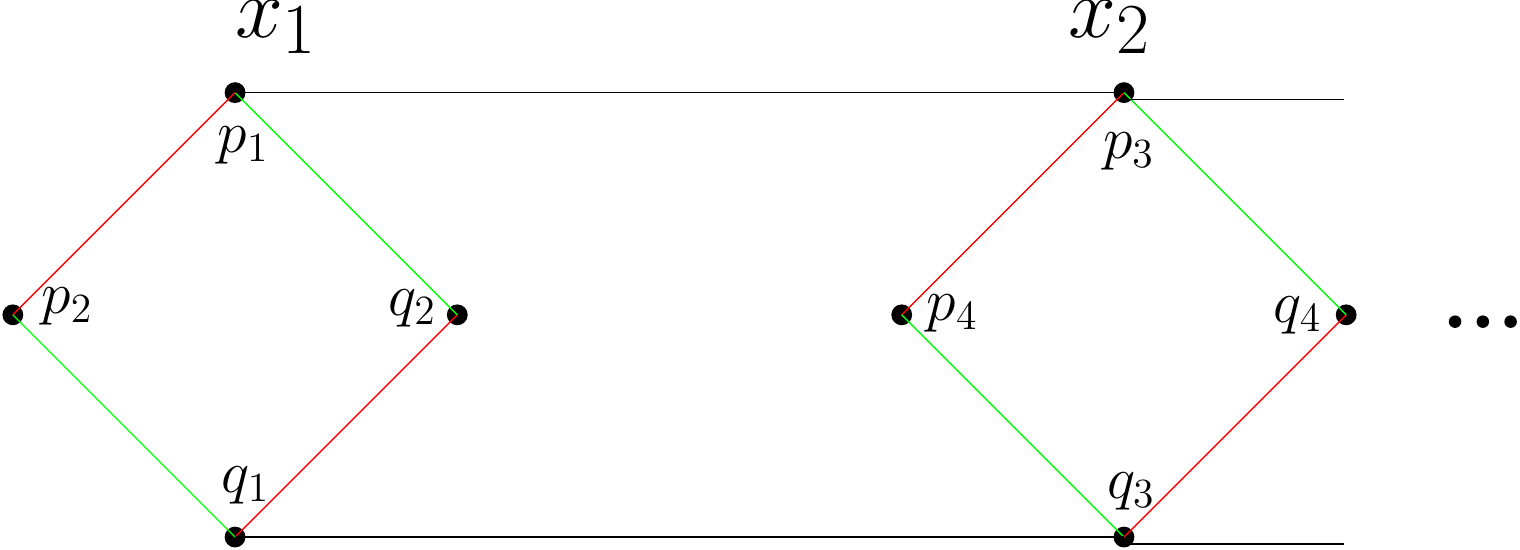}
\caption{Truth assignment.}
\label{fig:truth-assignment}
\end{figure}

A clause gadget consists of a configuration of 3 point pairs surrounding variable gadgets corresponding to the variables in that clause  (see Figure~\ref{fig:sum-hard}). The placement of the 3 point pairs depends on whether the literals appear positively or negatively. 

Consider clause $c_i$ which consists of variables $x_j$ and $x_k$. We create a pair of points $\{a_i, b_i\}$, each of which will be placed next to variable $x_j$. If $x_j$ appears negatively in $c_i$, then we place $a_i$ at distance 1 away from an endpoint of a green edge of $x_j$  and place $b_i$ at distance 1 away from the other endpoint of the green edge (see Figure~\ref{fig:sum-hard}c). If $x_j$ appears positively in $c_i$, then we do the same thing at the endpoints of a red edge (see Figure~\ref{fig:sum-hard}a). Then, we create a pair of points $\{d_i, e_i\}$ and follow the same procedure for variable $x_k$. Finally, for each clause gadget we create a pair of points $\{f_i, g_i\}$ and place them at a distance of 1 from certain variable gadget points chosen based on how many literals appear positively in clause $c_i$; zero, one, or two. The placement of $\{f_i, g_i\}$ for all three cases can be found in Figure~\ref{fig:sum-hard}.

As a technical note, to complete the construction, all clause gadget points placed around the same variable gadget vertex are separated from each other by distance 2, and these points are at distance only 1 from the nearest variable gadget vertex. This will ensure that the optimal solution will not link any two clause gadget points to each other. Also, note that we use the shortest path distance to define the weights of the rest of the edges in the graph.


Connecting all points except those associated with clause gadgets  into two MSTs has a base cost of $(2n - 2)L + 4n + 2 {n\choose2}M$. Now note that a clause evaluates to \emph{true} if and only if it costs 7 units to attach the clause gadget points to the backbones. A clause evaluates to \emph{false} if and only if it costs 9 units to attach the clause gadget points to the backbones (see Figure~\ref{fig:sum-hard}). Thus, it is now apparent that there exists a truth assignment in {\sc 2SAT} with $k$ clauses satisfied if and only if there exists a solution to the Min-Sum 2-MST problem with cost $(2n-2)L + 4n + 2{n \choose 2}M + 7k + 9(m-k)$.

\begin{figure}
\centering
\includegraphics[scale = 0.45]{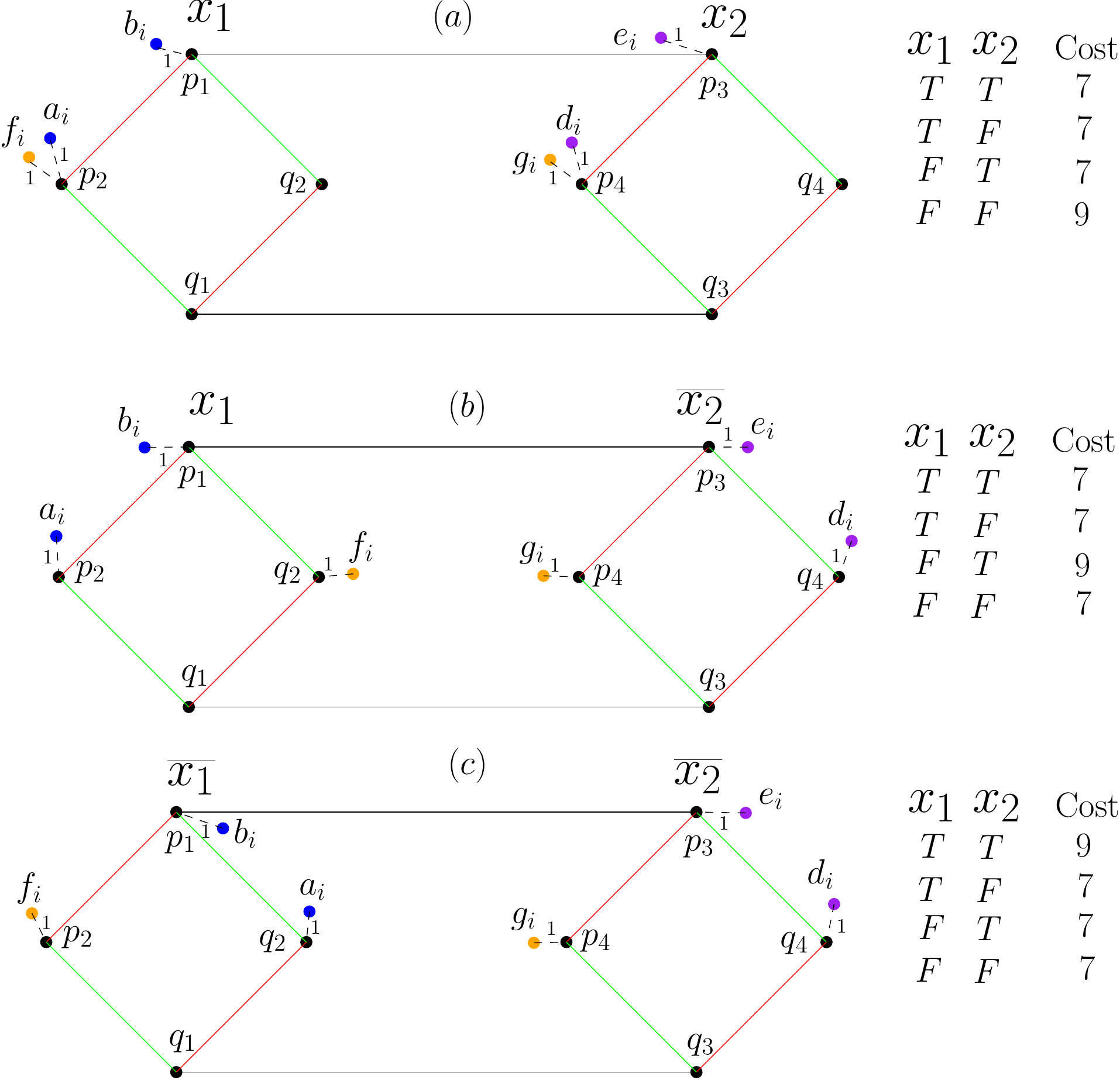}
\caption{Placement of clause gadget points and extra cost incurred to incorporate clause gadget points into two MSTs once a truth assignment over the variables is fixed.}
\label{fig:sum-hard}
\end{figure}
\end{proof}

\noindent{\textbf{Theorem~\ref{thm:minmax_m_apx}}.
\emph{The approximation algorithm for the Min-Sum 2-Matching problem serves as a 3-approximation for the Min-Max 2-Matching problem in  general metric spaces.}}

\begin{proof}
In this case we are concerned only with the larger of the two matchings returned by our approximation, $M(S_2)$, which, as described above, has weight bounded above by $|\hat{M}| + |M^*|$.  However, in this case, under the new cost function  $|\hat{M}| \le |OPT|$, and $|M^*| \le 2|OPT|$.  Therefore, $|M(S_2)| \le  |\hat{M}| + |M^*| \le 3 |OPT|.$  The example illustrated in Figure \ref{fig:2match_sum_tight} shows the approximation factor achieved by this algorithm is tight.
\end{proof}
\textbf{Theorem~\ref{thm:bottleneck_m}}.
\emph{There exists a 3-approximation to the Bottleneck 2-Matching Problem in general metric spaces.}[See section~\ref{apx:matching} of the appendix.] \\ \\
\textbf{Theorem~\ref{thm:minmax_2tsp}}.
\emph{There exists a $6 \beta$-approximation to the Min-Max 2-TSP problem, where $\beta$ is the approximation factor for TSP in a certain metric space.}

\begin{proof}
We use the same algorithm, and return the same coloring, $\mathcal{S} = R \cup B$, as in Section \ref{sec:tsp_min_sum}.  Let $APX$ be the coloring returned, and $|APX|$ be the cost of the larger TSP on both sets of $APX$. Let $OPT$ be the optimal solution. Note that $|APX| \leq \beta(|TSP(R)| + |TSP(B)|) \leq 3\beta(|TSP(R^*)| + |TSP(B^*)|) \leq 6\beta|OPT|$ as $|TSP(X_1)| + |TSP(X_2)| \leq 2max\{TSP(X_1), TSP(X_2)\}$ for any $X = X_1 \cup X_2$.
\end{proof}
\textbf{Theorem~\ref{thm:bottleneck_tsp}}.
\emph{There exists an 18-approximation algorithm for the Bottleneck 2-TSP problem.}

\begin{proof}
We remarked in section~\ref{sec:bmst} that there exists a 9-approximation to the problem of finding a partition that minimizes the weight of the bottleneck edge on two Hamilton paths built on the partition. A Hamilton path can be converted into a Hamilton cycle by at most doubling the weight of the bottleneck edge in the Hamilton path. This yields an 18-approximation to the Bottleneck 2-TSP problem.
\end{proof}

\section{Bottleneck 2-Matching}
\label{apx:matching}
We begin with a lemma concerning the structure of a feasible solution.  Let $\mathcal{S} = R \cup B$ be any feasible coloring to the Bottleneck 2-Matching instance. Let $M_B(X)$ be a minimum bottleneck matching on point set $X$.  Construct a graph $G = (V, E)$ where $V = S$ and $E = (M_B(R) \cup M_B(B) \cup (p_i, q_i)_{i = 1}^n)$ is the union of any pair of optimal bottleneck matchings on $R$, and $B$ and the edges $(p_i, q_i) \: \forall\: i$.

\begin{lemma}
$G$ is a 2-factor such that each input pair is contained in exactly one cycle, and each cycle contains an even number of input pairs.
\label{lem:BM_structure_lemma}
\end{lemma}

\begin{proof}
The edge set of $G$ is the union of two disjoint perfect matchings over $S$.  Therefore, each node has degree exactly 2 and $G$ is by definition a 2-factor.  Also, by definition of a 2-factor, each input point $p_i$ is part of a unique cycle, and in this case, as each node $p_i$ has an edge of the form $(p_i, q_i)$ incident to it, therefore, the point $q_i$ must be contained in the same cycle as $p_i \: \forall \: i$.  Thus, each input pair is contained in the same unique cycle in $G$.

By definition, two nodes defining an edge of $M_B(R) \cup M_B(B)$ must be of the same color, and two nodes defining an edge of the form $(p_i, q_i)$ must be of different color.   This together with the fact that the edges of this cycle alternate between the form $(p_i, q_i)$ and edges in $M_B(S_1) \cup M_B(S_2)$ implies that if we were to contract edges of the form $(p_i, q_i)$ we would still get a cycle which strictly alternates color.  This is only possible if there are the same number of red and blue nodes in the contracted cycle.  Thus an even number of input pairs.
\end{proof}

Using the above structure lemma we will argue that we can compute a graph with the same properties and extract a feasible partition with constant factor approximation guarantees.  Let $\hat{M}_B(\mathcal{S})$ be the minimum weight (exactly) one of a pair bottleneck matching over $\mathcal{S}$; note that edges of this matching go between points of $S$.  Let $M_B(S)$  be the minimum weight bottleneck matching over $S$ (excluding the edges $(p_i, q_i) \: \forall \:  i$).  Let $\hat{\lambda}$ (resp. $\lambda$) be the heaviest edge used in  $\hat{M}_B(\mathcal{S})$ (resp. $M_B(S)$) and let $\lambda^*$ be the heaviest edge in a minimum weight bottleneck matching computed over each of the two sets in $OPT$.  Note that $|\hat{\lambda}| \le |\lambda^*|$ and $|\lambda| \le |\lambda^*|$.

Begin by constructing a graph $G = (V = S, E = M_B(S) \cup (p_i, q_i)_{i = 1}^n)$, which is a 2-factor as its edge set is the union of two disjoint perfect matchings.  Note, it will be the case that each input pair exists in the same unique cycle.  If each cycle contains an even number of input pairs then this graph has the same structure as that described in Lemma \ref{lem:BM_structure_lemma} and thus we can extract a feasible partition from $G$. We will describe how to obtain this partition later. As $|\lambda^*| \ge |\lambda|$, this graph induces an optimal partition. On the other hand, if there exists a cycle with an odd number of input pairs (there must be an even number of such cycles) then we ``merge'' cycles of $G$ together into larger cycles until a point is reached in which each new ``super-cycle'' contains an even number of pairs. From this graph we can extract a constant factor approximation to an optimal coloring.

\begin{lemma}
If $G$ contains at least one cycle with an odd number of input pairs, then it is possible to merge cycles of $G$ into super-cycles, each of which contains an even number of input pairs, such that the heaviest edge $(
$excluding $(p_i, q_i) \; \forall \; i)$ in any super-cycle has weight at most $3|\lambda^*|$.
\end{lemma}

\begin{proof}
\textbf{(sketch)}
Superimpose a subset of the edges in $\hat{M}_B(\mathcal{S})$ over the nodes of $G$ in the following way. Consider only edges in $\hat{M}_B(\mathcal{S})$ which have endpoints in different cycles.  Treat each cycle in $G$ as a node and run Kruskal's algorithm until all of the aforementioned edges of $\hat{M}_B(\mathcal{S})$ are exhausted. This yields a forrest on the cycles of $G$. It is easy to see that every cycle of $G$ containing an odd number of input pairs has an edge of $\hat{M}_B(\mathcal{S})$ connecting it to some other cycle.  This implies that it is possible to merge all cycles which are connected by an edge of $\hat{M}_B(\mathcal{S})$ until one reaches a point where all cycles have an even number of pairs. We give a brief outline for this merging process.

 Find any maximal path $P$ in $G$ which alternates edges of the form $M_B(S)$ and $\hat{M}_B(\mathcal{S})$.  Consider the cycles in $G$ containing the edges of $M_B(S)$ in $P$ (see Figure \ref{fig:path}).
We will label the cycles in this path $C_1, C_2, \dots , C_k$ where $k$ is the number of cycles in the path.  We will ``stitch'' the cycles together into a final super-cycle by making connections between pairs of cycles with odd subscripts in sorted order then by making connections between cycles with even subscripts in reverse sorted order. Now, remove all edges of $P$ and we are left with a super-cycle (see Figure \ref{fig:stitch}).
Recall that the weight of each edge of $M_B(S)$ and $\hat{M}_B(\mathcal{S})$ is a lower bound on $|\lambda^*|$. Consider the two nodes that define some edge $e$ created in the stitching process. Note that $e$ can be replaced by path of at most three edges from $M_B(S) \cup \hat{M}_B(\mathcal{S})$. Thus, any edge not of the form $(p_i, q_i)$ in the super-cycle has weight at most $3|\lambda^*|$. It is not difficult to see that the stitching process can be done while keeping the weight of the bottleneck edge at most $3|\lambda^*|$ regardless of whether $k$ is even or odd.

\begin{figure}
\centering
\begin{subfigure}{4 in}
  \centering
  \includegraphics[scale = 0.5]{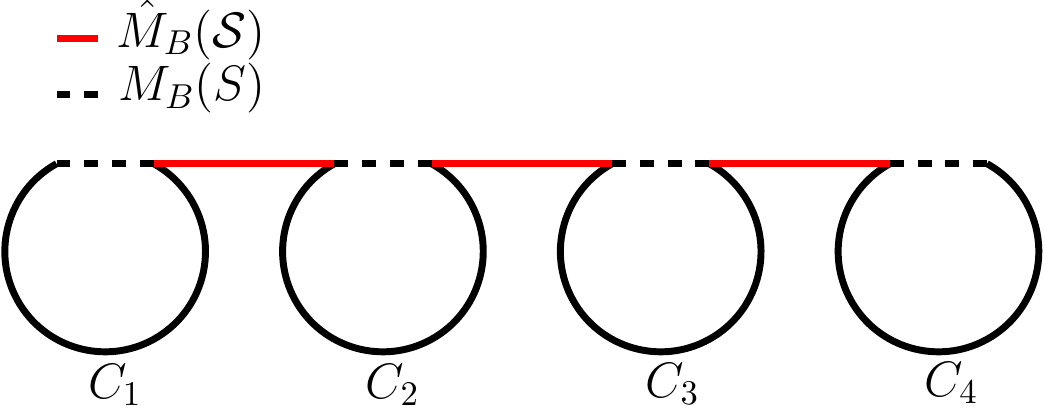}
  \caption{Before.}
  \label{fig:path}
\end{subfigure}
\begin{subfigure}{4 in}
  \centering
  \includegraphics[scale = 0.5]{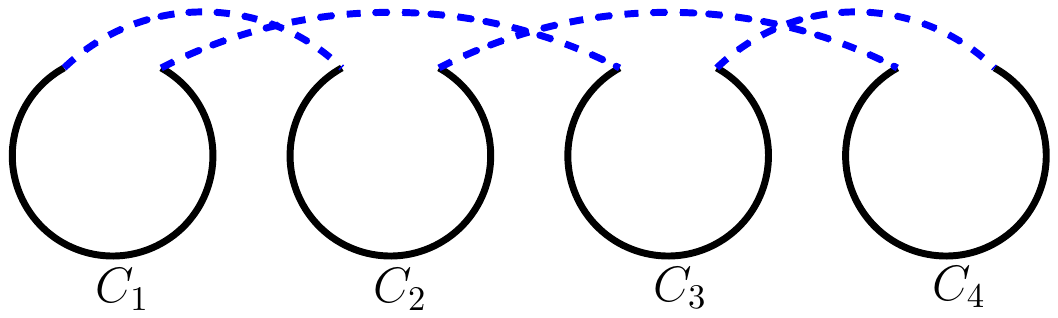}
  \caption{After.}
  \label{fig:stitch}
\end{subfigure}
\caption{Before and after stitching.}
\label{fig:path_and_stitch}
\end{figure}

It is possible that many maximal paths share an edge with the same cycle $C_i$ in $G$ (see Figure~\ref{fig:sc_before}). These edges must all be different because, when only considering edges of $M_B(S) \cup \hat{M}_B(\mathcal{S})$, the degree of each node in $G$ is at most two. This implies that any edge created in one stitching process is not altered by another stitching process and thus stitching processes are independent of one another. All cycles associated with these paths will be merged into the same super-cycle (see Figure~\ref{fig:sc_after}), and since the merging processes are independent, the bottleneck edge created is still of weight no larger than $3|\lambda^*|$.  
\end{proof}

\begin{figure}
\centering
\begin{subfigure}
  \centering
  \includegraphics[page = 1, width = 0.4\textwidth]{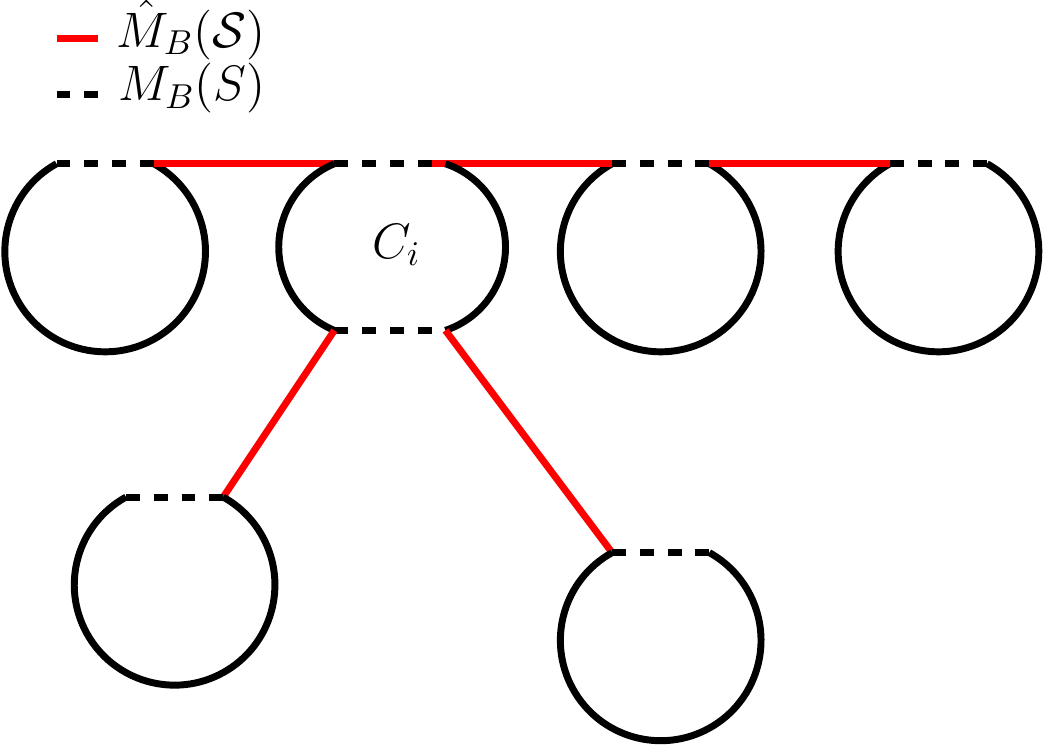}
  \caption{Before.}
  \label{fig:sc_before}
\end{subfigure}%
\begin{subfigure}
  \centering
  \includegraphics[page = 2, width = 0.4\textwidth]{super_cycle}
  \caption{After.}
  \label{fig:sc_after}
\end{subfigure}
\caption{Before and after merging into a super-cycle.}
\label{fig:sc}
\end{figure}

Now that each cycle contains an even number of input pairs, all that is left to show is how to create a feasible coloring from these cycles so that the weight of the heaviest edge in the bottleneck matching computed on either side of the partition is at most $3|\lambda^*|$. Notice that for each cycle, every other edge is of the form $(p_i, q_i)$. We will 2-color the nodes of each cycle red and blue so that if two nodes share an edge of the form $(p_i, q_i)$, they must be of different color, and if they share an edge not of the form $(p_i, q_i)$ they must be of the same color.  This coloring is clearly feasible and the weight of the heaviest edge in the matchings created on either side is equal to the weight of the heaviest edge not of the form $(p_i, q_i)$ among all of the cycles we have created; this edge has weight at most $3|\lambda^*|$.\\ \\
\textbf{Theorem~\ref{thm:bottleneck_m}}. 
\emph{There exists a 3-approximation to the Bottleneck 2-Matching problem in general metric spaces.} \\ \\

\end{document}